\documentclass[english]{article}
\usepackage[T1]{fontenc}
\usepackage[latin9]{inputenc}
\usepackage{geometry}
\usepackage{xcolor}
\usepackage{pdfcolmk}
\usepackage{amsthm}
\usepackage{amsmath}
\usepackage{amssymb}
\PassOptionsToPackage{normalem}{ulem}
\usepackage{ulem}

\makeatletter

\providecolor{lyxadded}{rgb}{0,0,1}
\providecolor{lyxdeleted}{rgb}{1,0,0}

\numberwithin{equation}{section}
\numberwithin{figure}{section}
\theoremstyle{plain}
\newtheorem{thm}{\protect\theoremname}
  \theoremstyle{remark}
  \newtheorem*{rem*}{\protect\remarkname}
  \theoremstyle{plain}
  \newtheorem{cor}[thm]{\protect\corollaryname}

\makeatother

\usepackage{babel}
  \providecommand{\corollaryname}{Corollary}
  \providecommand{\remarkname}{Remark}
\providecommand{\theoremname}{Theorem}

\begin{document}

\title{An exact reduction of the master equation to a strictly stable system
with an explicit expression for the stationary distribution}

\author{Daniel Soudry, Ron Meir {\normalsize }\\
{\normalsize Department of Electrical Engineering, }\\
{\normalsize the Laboratory for Network Biology Research, Technion,
Haifa , Israel}}
\maketitle
\begin{abstract}
The evolution of a continuous time Markov process with a finite number
of states is usually calculated by the Master equation - a linear
differential equations with a singular generator matrix. We derive
a general method for reducing the dimensionality of the Master equation
by one by using the probability normalization constraint, thus obtaining
a affine differential equation with a (non-singular) stable generator
matrix. Additionally, the reduced form yields a simple explicit expression
for the stationary probability distribution, which is usually derived
implicitly. Finally, we discuss the application of this method to
stochastic differential equations.
\end{abstract}

\section{Introduction}

Let $X\left(t\right)$ be a continuous time Markov process with discrete
states $\left\{ 1,2,...,M\right\} $, where $\ensuremath{1<M<\infty}$,
with $A_{ij}$ being the (non-negative) transition rate from state
$j$ to state $i$. We define $p_{i}\left(t\right)\in\left[0,1\right]$
to be the probability to be in state $i$ at time $t$, the probability
vector 
\begin{equation}
\ensuremath{{\bf \mathbf{p}}\left(t\right)\triangleq\left(p_{1}\left(t\right),...,p_{M}\left(t\right)\right)^{\top}}\in\left[0,1\right]^{M},\label{eq: p}
\end{equation}
and the rate matrix $\mathbf{A}$, so that 
\begin{equation}
\left(\mathbf{A}\right)_{ij}\triangleq\begin{cases}
A_{ij} & ,\,\mathrm{if}\,\, i\neq j\\
-\sum_{j\neq i}A_{ji} & ,\mathrm{\, if}\,\, i=j
\end{cases}\label{eq:A definition}
\end{equation}
and 
\begin{equation}
\frac{d\mathbf{p}\left(t\right)}{dt}=\mathbf{A}\mathbf{p}\left(t\right)\label{eq:master equation}
\end{equation}

is the corresponding master equation, with solution
\begin{equation}
\mathbf{p}\left(t\right)=\exp\left(\mathbf{A}t\right)\mathbf{p}\left(0\right)\,.\label{eq: master eq. solution}
\end{equation}

From the normalization of the probability, $\mathbf{p}\left(t\right)$
must be constrained at all time by 
\begin{equation}
{\bf e}^{\top}\mathbf{p}\left(t\right)=1\,\,\,\,;\,\,\,\,\mathbf{e}\triangleq\left(1,1,...,1\right)^{\top}.\label{eq:normalization}
\end{equation}
Note that from the properties of $\mathbf{A}$ (specifically, the
fact that $\mathbf{e}^{\top}\mathbf{A}=0$), if we start from an initial
condition $\mathbf{p}_{0}\in\left[0,1\right]^{M}$ so that $\mathbf{e}^{\top}\mathbf{p}_{0}=1$,
then, $\forall t$, $\mathbf{e}^{\top}\mathbf{p}\left(t\right)=1$
automatically - though this is not immediately obvious from the above
notation. 

In order to improve the interpretability of the above notation, we
combine Eq. \ref{eq:normalization} directly with Eq. \ref{eq:master equation}.
We shall henceforth assume that $X\left(t\right)$ is irreducible,
and reduce the dimensionality of the problem from $M$ to $M-1$ (section
\ref{sec:Reduction-of-the}). Note that if instead $X\left(t\right)$
is reducible with $K$ connected components, then the method suggested
here can be applied to each component separately, reducing the dimensionality
of the problem from $M$ to $M-K$ (see appendix \ref{sec:Generalization}).
The reduced form of the master equation (Eq. \ref{eq:reduce master eq 1}
or Eq. \ref{eq:reduce master eq 2}) has some {}``nice'' properties.
For example, in section \ref{sec:Properties-of} we prove that the
reduced form is strictly contracting; in section \ref{sec:Stationary-Distribution}
we show it is easy to find a novel explicit form for the stationary
(invariant) distribution using this reduced form (for the relation
with previous stationary distribution expressions see appendix \ref{sec:Equivalence-of-Stationary});
and in section \ref{sec:Application-to-Stochastic} we discuss the
application of this method to stochastic differential equations (SDE)
based on a population of independent Markov processes.

Note that similar reduction methods are rather popular for the special
case of a two state system $x\rightleftharpoons1-x$, in the context
of deterministic kinetic equations, which are the limit of the SDE
equations for an infinite population (e.g. \cite{Hodgkin1952}). In
a few special cases they were also used in SDE descriptions of specific
systems with more than one state \cite{Fox1994}.

\section{Reduction of the Master Equation\label{sec:Reduction-of-the}}

First, we make a few additional definitions:
\begin{enumerate}
\item $\mathbf{I}_{M}$ is the $M\times M$ identity matrix
\item $\mathbf{J}$ is $\mathbf{I}_{M}$ with it last row removed: $\mathbf{J}=\left(\begin{array}{cccccc}
1 & 0 & 0 & \cdots & 0 & 0\\
0 & 1 & 0 & \cdots & 0 & 0\\
\vdots & \vdots & \vdots & \ddots & 0 & 0\\
0 & 0 & 0 & \cdots & 1 & 0
\end{array}\right)$,~$\mathrm{dim}\left(J\right)=\left(M-1\right)\times M$
\item $\mathbf{e}_{M}\triangleq\left(0,0,...,1\right)^{\top}$
\item $\mathbf{H}\triangleq\left(\mathbf{I}_{M}-\mathbf{e}_{M}\mathbf{e}^{\top}\right)\mathbf{J}^{\top}=\left(\begin{array}{ccccc}
1 & 0 & 0 & \cdots & 0\\
0 & 1 & 0 & \cdots & 0\\
\vdots & \vdots & \vdots & \ddots & 0\\
0 & 0 & 0 & \cdots & 1\\
-1 & -1 & -1 & -1 & -1
\end{array}\right)$ , $\mathrm{dim}\left(H\right)=M\times\left(M-1\right)$
\item $\tilde{\mathbf{p}}\left(t\right)=\mathbf{J}\mathbf{p}\left(t\right)$,
$\mathrm{dim}\left(\tilde{\mathbf{p}}\right)=\left(M-1\right)\times1$
\end{enumerate}
Note that $\tilde{\mathbf{p}}\left(t\right)\in\left[0,1\right]^{M-1}$,
and the {}``hard'' normalization constraint has been lifted (instead
we remain with a {}``soft'' constraint ${\bf e}^{\top}\mathbf{J}^{\top}\mathbf{\tilde{p}}\left(t\right)\leq1$).
Using these definitions, we can use \ref{eq:normalization} to write
\begin{equation}
\mathbf{p}\left(t\right)=\mathbf{e}_{M}+\mathbf{H}\tilde{\mathbf{p}}\left(t\right)\label{eq: tilde p(t) 2 p(t)}
\end{equation}

Substituting this into Eq. \ref{eq:master equation} we obtain
\[
\mathbf{H}\frac{d\tilde{\mathbf{p}}\left(t\right)}{dt}=\mathbf{A}\mathbf{e}_{M}+\mathbf{A}\mathbf{H}\tilde{\mathbf{p}}\left(t\right)
\]

Multiplying this by $\mathbf{J}$ from the left, we obtain
\[
\mathbf{J}\mathbf{H}\frac{d\tilde{\mathbf{p}}\left(t\right)}{dt}=\mathbf{J}\mathbf{A}\mathbf{e}_{M}+\mathbf{J}\mathbf{A}\mathbf{H}\tilde{\mathbf{p}}\left(t\right)\,.
\]

Using the fact that 

\begin{equation}
\mathbf{J}\mathbf{H}=\mathbf{J}\left(\mathbf{I}-\mathbf{e}_{M}\mathbf{e}^{\top}\right)\mathbf{J}^{\top}\overset{}{=}\mathbf{J}\mathbf{J}^{T}=\mathbf{I}_{M-1}\,\,,\label{eq: JH}
\end{equation}
where we used $\mathbf{J}\mathbf{e}_{M}=0$ in the second equality.
Defining $\tilde{\mathbf{A}}\triangleq\mathbf{J}\mathbf{A}\mathbf{H},\tilde{\mathbf{b}}\triangleq\mathbf{J}\mathbf{A}\mathbf{e}_{M}$,
we can write our first reduced form of Eq. \ref{eq:master equation}
\begin{equation}
\frac{d\tilde{\mathbf{p}}\left(t\right)}{dt}=\tilde{\mathbf{b}}+\tilde{\mathbf{A}}\tilde{\mathbf{p}}\left(t\right)\,.\label{eq:reduce master eq 1}
\end{equation}

\section{Properties of $\tilde{\mathbf{A}}$\label{sec:Properties-of}}

Since\textbf{$\mathbf{A}$ }is a rate matrix of an irreducible process,
it has a single zero eigenvalue and all the other eigenvalues have
negative real parts \cite{Stirzaker2005}. Given this, we can find
the eigenvalues of $\tilde{\mathbf{A}}$. 
\begin{thm}
Assume $X\left(t\right)$ is an irreducible process, then $\tilde{\mathbf{A}}$
has the same eigenvalues as $\mathbf{A}$ - except its (unique) zero
eigenvalue.\label{thm:A eigenvalues}\end{thm}
\begin{proof}
To find the eigenvalues of $\tilde{\mathbf{A}}$, we examine the characteristic
polynomial 
\begin{eqnarray*}
\left|\tilde{\mathbf{A}}-\lambda\mathbf{I}_{M-1}\right| & = & \left|\mathbf{J}\mathbf{A}\mathbf{H}-\lambda\mathbf{I}_{M-1}\right|\\
 & \overset{\left(1\right)}{=} & \lambda^{M-1}\left|\lambda^{-1}\mathbf{J}\mathbf{A}\left(\mathbf{I}-\mathbf{e}_{M}\mathbf{e}^{\top}\right)\mathbf{J}^{\top}-\mathbf{I}_{M-1}\right|\\
 & \overset{\left(2\right)}{=} & \lambda^{M-1}\left|\lambda^{-1}\left(\mathbf{I}-\mathbf{e}_{M}\mathbf{e}^{\top}\right)\mathbf{J}^{\top}\mathbf{J}\mathbf{A}-\mathbf{I}_{M}\right|\\
 & \overset{\left(3\right)}{=} & \lambda^{-1}\left|\left(\mathbf{I}-\mathbf{e}_{M}\mathbf{e}^{\top}\right)\left(\mathbf{I}-\mathbf{e}_{M}\mathbf{e}_{M}^{\top}\right)\mathbf{A}-\lambda\mathbf{I}_{M}\right|\\
 & \overset{\left(4\right)}{=} & \lambda^{-1}\left|\mathbf{A}-\lambda\mathbf{I}_{M}\right|\\
 & \overset{\left(5\right)}{=} & \lambda^{-1}\prod_{i=1}^{M}\left(\lambda-\lambda_{i}\right)\\
 & = & \prod_{i=2}^{M}\left(\lambda-\lambda_{i}\right)
\end{eqnarray*}
where in $\left(1\right)$ we used the definition of $\mathbf{H}$
and the fact that $\left|\lambda\mathbf{X}\right|=\lambda^{M}\left|\mathbf{X}\right|$
for any $M\times M$ matrix and scalar $\lambda$, in$\left(2\right)$
we used Sylvester's determinant theorem  ($\left|I_{p}+\mathbf{B}\mathbf{C}\right|=\left|I_{p}+\mathbf{C}\mathbf{B}\right|$
for all $\mathbf{B}$, $\mathbf{C}$ matrices of size $p\times n$
and $n\times p$ respectively), in $\left(3\right)$ we used $\mathbf{J}^{\top}\mathbf{J}=\left(\mathbf{I}-\mathbf{e}_{M}\mathbf{e}_{M}^{\top}\right)$
and $\left|\lambda\mathbf{X}\right|=\lambda^{M}\left|\mathbf{X}\right|$
again, in $\left(4\right)$ we used $\mathbf{e}^{\top}\mathbf{e}_{M}=1$
and $\mathbf{e}^{\top}\mathbf{A}=0$ and in $\left(5\right)$ we denoted
by $\left\{ \lambda_{i}\right\} _{i=1}^{M}$ the eigenvalues of $\mathbf{A}$,
with $\lambda_{1}=0$. The last line concludes the proof.\end{proof}
\begin{rem*}
Although the eigenvalues of $\mathbf{A}$ and $\tilde{\mathbf{A}}$
are the same, their corresponding eigenvectors \textbf{$\mathbf{v}_{m}$
}and $\tilde{\mathbf{v}}_{m}$ are not tied by a simple projection,
namely \textbf{$\tilde{\mathbf{v}}_{m}\neq\mathbf{J}\mathbf{v}_{m}$.}
\end{rem*}
Recall again that a rate matrix \textbf{$\mathbf{A}$ }of an irreducible
process has a single zero eigenvalue and all the other eigenvalues
have negative real parts \cite{Stirzaker2005}. Using theorem \ref{thm:A eigenvalues}
this immediately gives
\begin{cor}
$\tilde{\mathbf{A}}$ is a stable matrix - i.e. all its eigenvalues
have a strictly negative real part.\label{cor: A stable}
\end{cor}
Specifically, since $\tilde{\mathbf{A}}$ does not have any zero eigenvalues,
\begin{cor}
$\tilde{\mathbf{A}}$ is a non-singular matrix, and therefore, invertible.\label{cor: A non-singular}
\end{cor}

\section{Stationary Distribution\label{sec:Stationary-Distribution}}

Recall (\cite{Stirzaker2005}) that if $X\left(t\right)$ is irreducible
then $\mathbf{p}\left(t\right)\rightarrow\mathbf{p}_{\infty}$, a
stationary distribution which is the (unique) zero eigenvector of
the matrix $\mathbf{A}$, 

\begin{equation}
0=\mathbf{A}\mathbf{p}_{\infty}\,.\label{eq:steady state eq}
\end{equation}
This is an \emph{implicit} equation for $\mathbf{p}_{\infty}$. However,
using the our reduced version, it is easy to find an explicit expression
for the stationary distribution .

Using Eq. \ref{eq:reduce master eq 1} and Corollary \ref{cor: A non-singular},
we define 
\begin{equation}
\tilde{\mathbf{p}}_{\infty}\triangleq-\tilde{\mathbf{A}}^{-1}\tilde{\mathbf{b}}\label{eq: tilde p_inf}
\end{equation}

and re-write Eq. \ref{eq:reduce master eq 1} as 
\begin{equation}
\frac{d\tilde{\mathbf{p}}\left(t\right)}{dt}=\tilde{\mathbf{A}}\left(\tilde{\mathbf{p}}\left(t\right)-\tilde{\mathbf{p}}_{\infty}\right)\,,\label{eq:reduce master eq 2}
\end{equation}

which is our second reduced form of Eq. \ref{eq:master equation}.

Since \textbf{$\tilde{\mathbf{A}}$ }is stable, $\mathbf{p}\left(t\right)\rightarrow\tilde{\mathbf{p}}_{\infty}$,
and so the solution of \ref{eq:reduce master eq 2} is 
\[
\tilde{\mathbf{p}}\left(t\right)=\tilde{\mathbf{p}}_{\infty}+\left(\tilde{\mathbf{p}}\left(0\right)-\tilde{\mathbf{p}}_{\infty}\right)e^{\tilde{\mathbf{A}}t}\,.
\]

And so, we found an explicit expression for the steady state distribution
in the reduced form
\[
\tilde{\mathbf{p}}_{\infty}=-\left(\mathbf{JAH}\right)^{-1}\mathbf{J}\mathbf{A}\mathbf{e}_{M}\,\,.
\]

Returning to the original form, using Eq. \ref{eq: tilde p(t) 2 p(t)},
we obtain the explicit expression 
\begin{equation}
\mathbf{p}_{\infty}=\left(\mathbf{I}_{M}-\mathbf{H}\left(\mathbf{JAH}\right)^{-1}\mathbf{J}\mathbf{A}\right)\mathbf{e}_{M}\,\,.\label{eq: p_inf 1}
\end{equation}

In section \ref{sec:Equivalence-of-Stationary} we compare this expression
with previous results. Note that for a discrete time Markov chain
with transition matrix $\mathbf{P}$, we can again find the stationary
distribution by substituting $\mathbf{A}=\mathbf{I}-\mathbf{P}$ in
either Eq. \ref{eq: p_inf 1} or \ref{eq: p_inf 2}.

\section{The reduction methods in stochastic differential equations\label{sec:Application-to-Stochastic}}

Consider a population of identical, irreducible and independent Markov
processes $\left\{ X_{n}\left(t\right)\right\} _{n=1}^{N}$, where
each process has states $\left\{ 1,2,...,M\right\} $, where $\ensuremath{1<M<\infty}$.
Also, for all processes, $A_{ij}$ is the transition rate from state
$j$ to state $i$, and $\mathbf{A}$ is the corresponding matrix.
We denote by $x_{i}\left(t\right)$ the fraction of processes that
are in state $i$ at time $t$ (not following convention of using
upper case only for random variables). Formally 
\[
x_{i}\left(t\right)\triangleq\frac{1}{N}\sum_{n=1}^{N}\mathcal{I}\left[X_{n}\left(t\right)=i\right]\,,
\]
where $\mathcal{I}\left[\cdot\right]$ is the indicator function.
Also, we denote $\mathbf{x}=\left(x_{1},...,x_{M}\right)^{\top}$.
From normalization,
\begin{equation}
{\bf e}^{\top}\mathbf{x}\left(t\right)=1\,\,\,\,;\,\,\,\,\mathbf{e}\triangleq\left(1,1,...,1\right)^{\top}.\label{eq:normalization-1}
\end{equation}

As derived in \cite{OrioSoudry2011}, for large enough $N$ we can
approximate the dynamics of \textbf{$\mathbf{x}$ }by the following
$n-$dimensional stochastic differential equation (SDE)
\begin{equation}
\dot{\mathbf{x}}\left(t\right)=\mathbf{A}\mathbf{x}\left(t\right)+\mathbf{B}\left(\mathbf{x}\left(t\right)\right)\xi\left(t\right)\label{eq: SDE 1}
\end{equation}
where $\xi$ is a vector of $M\left(M-1\right)/2$ independent white
noise processes with zero mean and correlation $\left\langle \xi\left(t\right)\xi\left(t'\right)\right\rangle =\delta\left(t-t'\right)$
($\left\langle \cdot\right\rangle $ denotes ensemble expectation),
and \textbf{$\mathbf{B}$ }is a (sparse) $M\times M\left(M-1\right)/2$
matrix, with
\[
B_{ik}=\frac{1}{\sqrt{N}}\mathrm{sgn}\left(i-m_{ik}\right)\sqrt{A_{im_{ik}}x_{m_{ik}}+A_{m_{ik}i}x_{i}}
\]

where $k$ is the index of a transition pair ($i\rightleftharpoons j$)
and $m_{ik}$ is index of the state connected to state $i$ by transition
pair $k$. Note that since $N$ is large, any Ito correction would
be of size $O\left(N^{-2}\right)$, and is therefore neglected here.

We can reduce the form of Eq. \ref{eq: SDE 1} using \ref{eq:normalization-1}
in a similar way as we did for the Markov process. Defining $\tilde{\mathbf{A}}=\mathbf{J}\mathbf{A}\mathbf{H}$
(as before), $\tilde{\mathbf{B}}=\mathbf{J}\mathbf{B}$ (with $x_{K}$
replaced by $1-x_{1}-x_{2}...-x_{K-1}$) and $\tilde{\mathbf{x}}_{\infty}\triangleq\tilde{\mathbf{p}}_{\infty}=\left(\tilde{\mathbf{A}}\right)^{-1}\mathbf{J}\mathbf{A}\mathbf{e}_{M}$,
we obtain the following equation for the reduced state vector $\tilde{\mathbf{x}}=\mathbf{J}\mathbf{x}$
\begin{equation}
\frac{d\tilde{\mathbf{x}}\left(t\right)}{dt}=\mathbf{\tilde{A}}\left(\tilde{\mathbf{x}}\left(t\right)-\tilde{\mathbf{x}}_{\infty}\right)+\tilde{\mathbf{B}}\left(\tilde{\mathbf{x}}\left(t\right)\right)\boldsymbol{\xi}\left(t\right)\,.\label{eq:SDE 2}
\end{equation}
As before $\tilde{\mathbf{A}}$ is a stable matrix. Additionally,
the reduced diffusion matrix $\tilde{\mathbf{D}}\triangleq\tilde{\mathbf{B}}\tilde{\mathbf{B}}^{\top}$
is positive definite (in contrast to $\mathbf{D}=\mathbf{B}\mathbf{B}^{\top}$,
which is only semi-definite). This stems from the combination of the
following facts: (1) $\tilde{\mathbf{D}}=\tilde{\mathbf{B}}\tilde{\mathbf{B}}^{\top}$
is symmetric (2) The rank of $\tilde{\mathbf{B}}$ is $M-1$ (for
irreducible $X_{n}\left(t\right)$) (3) For any real matrix $\mathbf{X}$,
$\mathrm{rank}\left(\mathbf{X}\mathbf{X}^{\top}\right)=\mathrm{rank}\left(\mathbf{X}\right)$
\cite{Brookes2005}.

\appendix

\part*{Appendix}

\section{Generalization to a reducible processes\label{sec:Generalization}}

Assume now that $X\left(t\right)$ is a reducible process, with $K$
connected components $C_{k}$, $k=\left\{ 1,2,...,K\right\} $, where
$C_{k}$ contains $M^{\left(k\right)}$ states. In this case, we can
write 
\[
\mathbf{A}=\left(\begin{array}{cccc}
\mathbf{A}^{\left(1\right)} & 0 & \cdots & 0\\
0 & \mathbf{A}^{\left(2\right)} & \cdots & 0\\
\vdots & \vdots & \ddots & \vdots\\
0 & 0 & \cdots & \mathbf{A}^{\left(K\right)}
\end{array}\right)\,.
\]
Also, the normalization condition (Eq. \ref{eq:normalization}) can
be expanded to each component separately, 
\[
\forall k:\mathbf{u}_{k}^{\top}\mathbf{p}\left(t\right)=q_{k}\,\,\,;\,\,\,\left(\mathbf{u}_{k}\right)_{m}\triangleq\mathcal{I}\left[m\in C_{k}\right]
\]
where $\sum_{k}q_{k}=1$. In order to derive the reduced form of Eq.
\ref{eq:master equation} in this case, we just have to find the reduced
form for each component separately, and then concatenate the equations,
reducing the dimensionality from $M$ to $M-K$. Formally,we define:
\begin{enumerate}
\item $a_{k}$ is the index of the last ($M^{\left(k\right)}$ ) state in
$C_{k}$.
\item $\mathbf{L}$ is $\mathbf{I}_{M}$ with the rows corresponding to
$\left\{ a_{k}\right\} _{k=1}^{K}$ removed.
\item $\mathbf{f}$ is an length-$M$ vector for which all the indices $\left\{ a_{k}\right\} _{k=1}^{K}$
equal $q_{k}$ and all the rest equal $0$.
\item $\mathbf{H}_{m}$ as \textbf{$\mathbf{H}$} with $M=m$. 
\item $\mathbf{G}=\left(\begin{array}{cccc}
\mathbf{H}_{M^{\left(1\right)}} & 0 & \cdots & 0\\
0 & \mathbf{H}_{M^{\left(2\right)}} & \cdots & 0\\
\vdots & \vdots & \ddots & \vdots\\
0 & 0 & \cdots & \mathbf{H}_{M^{\left(K\right)}}
\end{array}\right)$
\item $\tilde{\mathbf{p}}\left(t\right)=\mathbf{J}\mathbf{p}\left(t\right)$
\end{enumerate}
Using these definitions, we can use \ref{eq:normalization} to write
\begin{equation}
\mathbf{p}\left(t\right)=\mathbf{f}+\mathbf{G}\tilde{\mathbf{p}}\left(t\right)\,.\label{eq: tilde p(t) 2 p(t)-1}
\end{equation}

Substituting this into Eq. \ref{eq:master equation} we obtain
\[
\mathbf{G}\frac{d\tilde{\mathbf{p}}\left(t\right)}{dt}=\mathbf{A}\mathbf{f}+\mathbf{A}\mathbf{G}\tilde{\mathbf{p}}\left(t\right)\,.
\]

Multiplying this by $\mathbf{J}$ from the left, we obtain
\[
\mathbf{L}\mathbf{\mathbf{G}}\frac{d\tilde{\mathbf{p}}\left(t\right)}{dt}=\mathbf{L}\mathbf{A}\mathbf{f}+\mathbf{L}\mathbf{A}\mathbf{\mathbf{G}}\tilde{\mathbf{p}}\left(t\right)\,.
\]

Using the fact that 

\begin{equation}
\mathbf{L}\mathbf{G}=\mathbf{I}_{M-K}\,\,,\label{eq: LG}
\end{equation}
and defining $\tilde{\mathbf{A}}\triangleq\mathbf{L}\mathbf{A}\mathbf{G},\tilde{\mathbf{b}}\triangleq\mathbf{L}\mathbf{A}\mathbf{f}$,
we can write our first reduced form of Eq. \ref{eq:master equation}
\begin{equation}
\frac{d\tilde{\mathbf{p}}\left(t\right)}{dt}=\tilde{\mathbf{b}}+\tilde{\mathbf{A}}\tilde{\mathbf{p}}\left(t\right)\,.\label{eq:reduce master eq 3}
\end{equation}

which has dimension $M-K$. All the other results we derived for the
irreducible case (i.e. the properties of $\tilde{\mathbf{A}}$, the
stationary distribution, etc.) can be similarly proven.

\section{Relations to previous results - stationary distribution expression\label{sec:Equivalence-of-Stationary}}

In the main text (Eq. \ref{eq: p_inf 1}) we derived an expression
for the stationary distribution
\begin{equation}
\mathbf{p}_{\infty}=\left(\mathbf{I}_{M}+\mathbf{H}\left(\mathbf{JAH}\right)^{-1}\mathbf{J}\mathbf{A}\right)\mathbf{e}_{M}\,\,.\label{eq:p_inf 1 1}
\end{equation}

Note however, that this is not the first explicit form suggested for
the solution of Eq. \ref{eq:steady state eq}. For example, \cite{Paige1975}
proved that
\begin{equation}
\mathbf{p}_{\infty}=\left(\mathbf{A}+\mathbf{v}\mathbf{e}^{\top}\right)^{-1}\mathbf{v}\label{eq: p_inf 2}
\end{equation}
for any $\mathbf{v}$ such that $\mathbf{e}^{\top}\mathbf{v}\neq0$.

Both Eq. \ref{eq: p_inf 1} and Eq. \ref{eq: p_inf 2} must be equal
and behave similarly if we vary $\mathbf{A}$. For example, Eq. \ref{eq:p_inf 1 1}
immediately implies that $\mathbf{p}_{\infty}$ does not change if
we scale $\mathbf{A}\rightarrow c\mathbf{A}$ by some non-zero constant,
as implied by Eq. \ref{eq:steady state eq}. This can be seen also
in Eq. \ref{eq: p_inf 2} if we scale $\mathbf{v}\rightarrow c\mathbf{v}$
simultaneously with the scaling in \textbf{$\mathbf{A}$}.\textbf{ }

To prove that both equations coincide (for any choice of $\mathbf{v}$),
we equate them, expecting to derive an identity:
\begin{eqnarray*}
\mathbf{v} & = & \left(\mathbf{A}+\mathbf{v}\mathbf{e}^{\top}\right)\left(\mathbf{I}_{M}+\mathbf{H}\left(\mathbf{JAH}\right)^{-1}\mathbf{J}\mathbf{A}\right)\mathbf{e}_{M}\\
 & = & \mathbf{A}\mathbf{e}_{M}+\mathbf{A}\mathbf{H}\left(\mathbf{JAH}\right)^{-1}\mathbf{J}\mathbf{A}\mathbf{e}_{M}+\mathbf{v}\mathbf{e}^{\top}\mathbf{e}_{M}+\mathbf{v}\mathbf{e}^{\top}\mathbf{H}\left(\mathbf{JAH}\right)^{-1}\mathbf{J}\mathbf{A}\mathbf{e}_{M}
\end{eqnarray*}

Since $\mathbf{e}^{\top}\mathbf{e}_{M}=1$ and $\mathbf{e}^{\top}\mathbf{H}=0$,
we obtain
\[
0=\mathbf{A}\mathbf{e}_{M}+\mathbf{A}\mathbf{H}\left(\mathbf{JAH}\right)^{-1}\mathbf{J}\mathbf{A}\mathbf{e}_{M}
\]

multiplying this by $\mathbf{J}$ from the left we get $0=0$, as
expected. Multiplying by $\mathbf{e}^{\top}$ from the left also gives
$0=0$, since $\mathbf{e}^{\top}\mathbf{A}=0$. Since the row vectors
of \textbf{$\mathbf{J}$}, combined with $\mathbf{e}^{\top}$, span
the vector space $\mathbb{R}^{M}$, this concludes our proof.

\bibliographystyle{plain}

\end{document}